\documentclass[10pt, conference, compsocconf]{IEEEtran}

\usepackage{cite}
\usepackage[usenames,dvipsnames]{xcolor}
\usepackage{amsmath}
\usepackage{graphicx}
\usepackage{multirow}
\usepackage{subfigure}
\usepackage{url}
\usepackage{array,booktabs}
\usepackage{xspace}

\usepackage{hyperref}

\usepackage{amsthm}
\theoremstyle{definition}
\newtheorem{definition}{Definition}
\theoremstyle{remark}

\newtheorem{claim}{Claim} 

\usepackage{algorithm}
\usepackage[noend]{algpseudocode}
\algrenewcommand\textproc{}
\algdef{SE}[EVENT]{Event}{EndEvent}[1]{\textbf{upon event}\ #1\ \algorithmicdo}{\algorithmicend\ \textbf{event}}%
\algtext*{EndEvent}

\usepackage{caption}
\newcommand{\remove}[1]{}

\begin{document}

\newcommand{\lidar}{LiDAR}
\newcommand{\dbscan}{DBScan}
\newcommand{\PCLeCluster}{\textit{PCL\_E\_Cluster}}
\newcommand{\lisco}{\textit{Lisco}}
\newcommand{\minpts}{\textit{minPts}}
\newcommand{\ford}{\textit{Ford Campus}}

\title{\lisco{}: A Continuous Approach in \lidar{} Point-cloud Clustering}

\author{\IEEEauthorblockN{Hannaneh Najdataei, Yiannis Nikolakopoulos, Vincenzo Gulisano, Marina Papatriantafilou}
\IEEEauthorblockA{Chalmers University of Technology
\\Gothenburg, Sweden\\
\{hannajd,ioaniko,vinmas,ptrianta\}@chalmers.se}
}

\maketitle

\begin{abstract}
The light detection and ranging (\lidar{}) technology allows to sense surrounding objects with fine-grained resolution in a large areas.
Their data (aka point clouds), generated  continuously at very high rates, can provide information to support automated functionality in cyberphysical systems.  
Clustering of point clouds is a key problem to extract this type of information.
Methods for solving the problem
in a continuous fashion can facilitate improved processing in e.g. fog architectures, allowing continuous, streaming processing of data close to the sources.
We propose \lisco{}, a single-pass continuous Euclidean-distance-based clustering of \lidar{} point
clouds, that maximizes the granularity of the data
processing pipeline. Besides its algorithmic analysis, we provide a thorough experimental evaluation and highlight its up to 3x improvements and its scalability benefits compared to the baseline, using
both real-world datasets as well as synthetic ones to fully
explore the worst-cases.

\end{abstract}


\IEEEpeerreviewmaketitle


\section{Introduction}\label{sec:introduction}

Active sensors that are able to measure properties of the surrounding environment with very fine-grained time resolution are utilized more and more in cyber-physical systems, such as autonomous vehicles, digitalized automated industrial environments and more.
These sensors can produce large streams of readings, with the \lidar{} (light detection and ranging) sensor being a prominent example. A \lidar{} sensor commonly mounts several lasers on a rotating column; at each rotation step, these lasers shoot light rays and, based on the time the reflected rays take to reach back the sensor, they produce a stream of distance readings at high rates, in the realm of million of readings per second.

As common in big data applications, one of the challenges in leveraging the information carried by such large streams is the need for efficient methods that can rapidly distill the valuable information from the raw measurements~\cite{Gulisano2015debs,fu2014bigdata,Nikolakopoulos2016ericsson,Zacheilas2017debs}.
A common problem in the analysis of the \lidar{} sensor data is 
\remove{segmentation, which can be addressed through} \emph{clustering} of the raw distance measurements, in order to detect objects surrounding the sensor~\cite{rusu20113d}.
This can, for instance, enable the detection of surrounding obstacles and prevent accidents (e.g. avoiding pedestrians in autonomous driving) or study the motion feasibility of objects in factories' production paths~\cite{FCC-report}. 


\subsection*{Challenges and contributions}

The processing time incurred by clustering of raw measurements (aka \emph{point clouds}) represents one of the main challenges in this context because of the high rates and the need for the clustering outcome to be available in a timely manner for it to be useful. Furthermore, the accuracy of the clustering is challenging as well, since readings from objects that are at different distances from the sensor can vary a lot in density. 

A key performance enabler for high-rate data streaming analysis is the pipelining of its composing tasks.
Nevertheless, common state-of-the art approaches for clustering of \lidar{} data (cf \cite{woo2002new},\cite{rusu2009close},\cite{vo2015octree} and more detail in \autoref{sec:related} elaborating on related work) first organize the points to be clustered (e.g. sorting them so that points close in space are also close in the data structure maintaining them) and only then perform the clustering by querying the organized data (e.g. by running neighbor-query as discussed in \autoref{sec:preliminaries}).
By doing this, they introduce a batch based processing that affects the clustering performance.

To overcome this, we target a single-pass analysis that will enable fine-grained pipelining in processing the data.
We propose a new method that achieves \lidar{} data-point clustering, called \lisco{}, that allows to boost processing throughput by maximizing the internal pipelining of the analysis steps, through a key idea that can exploit the inner ordering of the data generated by a \lidar{} sensor.


\remove{ 
In the context of clustering of \lidar{} data, the pipelining is prevented by the supporting data structures used by state-of-the-art implementations.
Specifically, by the sorting data structures that are first populated with the points and later queries to identify neighbors.

The key intuition behind \lisco{} that allows to boost performance by maximizing the internal pipelining of the clustering is to exploit the fact that data comes already sorted (to a certain extent) and the search can be made without external support (which would break the one-pass requirement).
} 

In more detail, we make the following contributions:
\begin{enumerate}
    \item We introduce \lisco{}, a new algorithm for Euclidean-distance-based clustering of \lidar{} point clouds, that maximizes the granularity of the data processing pipeline, without the need for supporting sorting data structures. \item We provide a fully implemented prototype and we discuss \lisco{}'s complexity in connection to the state-of-the-art Euclidean-distance-based clustering method in the Point Cloud Library (PCL), which we adopt as baseline due to its known accuracy, efficiency and wide use-base \cite{rusu20113d}.\footnote{Available through: \texttt{http://pointclouds.org}}
    \item We perform a thorough comparative evaluation, using both real-world datasets as well as synthetic ones to fully explore the worst-cases and the spectrum of trade-offs. We achieve a significant improvement, up to 3 times faster than the baseline and we also show significant scalability benefits.
\end{enumerate}

The rest of the paper is organized as follows. \autoref{sec:preliminaries}~overviews the \lidar{} sensor, the data it produces and clustering-related techniques that exist for such data.
\autoref{sec:lisco}~presents the main idea, the outline and argues about the properties of the proposed \lisco{} method, while the algorithmic implementation is given in \autoref{sec:implementation}. 
We evaluate the proposed method in \autoref{sec:evaluation}. Finally, we discuss related work and conclude in \autoref{sec:relatedwork} and \autoref{sec:conclusions}, respectively.

\section{Preliminaries}
\label{sec:preliminaries}

In this section, we give details about the key properties of \lidar{} sensors and the data they generate. 
We also provide a detailed problem description and evaluation criteria of solutions. Finally, we briefly describe the Euclidean-distance based clustering method in PCL, which we use as baseline as explained in the introduction.

\subsection{\lidar{} - sensor and data}
\label{sec:lidar}

The light detection and ranging (\lidar{}) technology allows sensing surrounding objects with fine-grained resolution in a large areas.

The \lidar{} sensor mounts $L$ lasers in a column, each measuring the distance from the target by means of the time difference from emitted and reflected light pulses.
The column of lasers performs $R$ rotations per second, each consisting of $S$ steps, producing a set of $n$ points, also called \textit{point cloud}.
The number of points reported by the \lidar{} sensor for each rotation can be lower than $L \times S$, since some of the emitted pulses might not hit any obstacle.
We refer to the angle in the x-y plane between two consecutive steps\footnote{If the horizontal angle between steps is not constant and lasers  are not perfectly aligned in the x-y plane, then $\Delta \alpha$ refers to the minimum such angle. Similarly, if the vertical angle between lasers is not constant, $\Delta \theta$ refers to the minimum such angle.} as $\Delta \alpha$ (e.g. measured in \emph{radians}) and to the elevation angle from the x-y plane between two consecutive lasers as $\Delta \theta$.
Each point $p$ is described through attributes $\left<d,l,s\right>$ where $d$, $l$, $s$ are the measured distance, the laser index and the step index.
The measured distance, $d$, is a value greater than or equal to 0. Value 0 for $d$ shows no reflection in the point.
In the following, we use the notation $p_x$ to refer to attribute $x$ of point $p$ (i.e., $p_d$ refers to the distance reported for point~$p$).

\begin{figure*}[ht!]
  \centering
  \includegraphics[width=0.8\linewidth]{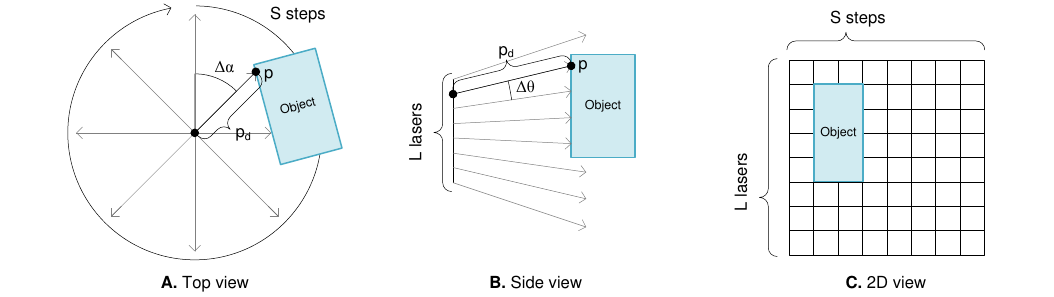}
  \caption{Top and side views for the \lidar{}'s emitted light pulses showing steps and lasers, together with the resulting 2D view. In the 2D view, non-reflected pulses are white while reflected ones are coloured.}
  \label{fig:lidar}
\end{figure*}

Let us consider in the following examples that $L=8$ and $S=8$. We visually present the points by unfolding them as a 2D matrix, where each column contains $L$ rows and each rotation step is a column. We assume that new data is delivered from the physical sensor with the granularity of one rotation step (Figure \ref{fig:lidar}).

\subsection{Problem formulation: from point clouds to clusters}
\label{sec:clustering}

Given a set of points corresponding to \lidar{} measurements, we want to identify disjoint groups of them that can be potential objects in the environment surrounding  the sensor. A natural criterion, commonly used in the literature and applications, is the distance between points. In particular, adopting the problem definition  in \cite{rusu2010semantic} (Chapter 4), which we paraphrase here for ease of reference:

\theoremstyle{definition} 
\begin{definition}\label{def:problem}
[\emph{Euclidean-distance based clustering}\/] Given $n$ points in 3D
space, we seek to partition them into some (unknown) number of clusters $C$ using the Euclidean-distance metric, such that every cluster contains at least a predefined number of points ($minPts$), that is $\forall j, |C_j| \geq minPts$, and all clusters are disjoint, that is $C_i \cap C_j = \emptyset, \forall i \neq j$.
Two points $p_i$ and $p_j$ should be clustered together if their
Euclidean distance $||p_i-p_j||_2$  is at most $\epsilon$, with $\epsilon$ being a predefined threshold. 
\end{definition}

To facilitate the presentation of the baseline algorithm and our proposed one, we introduce the
$\epsilon$\textit{-neighborhood of a point} $p$: the set of points of the input whose Euclidean distance from $p$ is at most $\epsilon$.
A set of points closed under the union of their  $\epsilon$-neighborhoods is characterized as \textit{noise} if its cardinality is less than $minPts$.

It should be noticed that, when clustering data from scenarios like the vehicular one, a pre-processing task is usually defined to filter out points that refer to the ground, since many objects laying on it would be otherwise clustered together. Since ground removal can be implemented as a non expensive and continuous filtering operation \cite{himmelsbach2010fast} (e.g. by removing points below a certain threshold, as we do in \autoref{sec:evaluation}) we do not further discuss it in the remainder. 

\subsection{Euclidean-distance based clustering in PCL}
\label{sec:pcl}

PCL~\cite{rusu20113d} provides a set of tools based on a collection of state-of-the-art algorithms to process 3D data, including filtering, clustering, surface reconstruction and more.
In this section we review its method for cluster extraction, which is an Euclidean-distance based clustering and we use it as a baseline. For brevity we call this algorithm \PCLeCluster{} in the rest of this document.

\PCLeCluster{} works on  batches of data points.
It first builds a kd-tree to facilitate finding the nearest neighbors  of points. 
Subsequently, it proceeds as described by algorithm \ref{alg:pcl} to extract the clusters.

\begin{algorithm}
\small
\caption{Main loop  of \PCLeCluster{}}\label{alg:pcl}
\begin{algorithmic}[1]
\State clusters = $\emptyset$
\For{$p \in$ P}
    \State Q = $\emptyset$
    \If{$p${.status $\neq$ processed}}
        \State Q.\texttt{add}($p$)  \label{pcl:add}
        \For{$q \in$ Q} 
            \State $q${.status = 'processed'}
            \State $N$ = \texttt{GetNeighbors}($q,\epsilon$)
            \State Q.\texttt{addAll}($N$) 
        \EndFor
        \If{$\text{size(Q)} \geq \text{\minpts{}}$}
            clusters.\texttt{add}(Q)
        \EndIf
    \EndIf
\EndFor
\end{algorithmic}
\end{algorithm}

Starting from any arbitrary unprocessed point $p$, the algorithm adds it in an empty list, $Q$ (line \ref{pcl:add}). Then, \PCLeCluster{} adds all the points of the $\epsilon$-neighborhood of each member of Q to it.
After processing all members of the list $Q$, if its size is greater than or equal to $\minpts{}$, the list is returned a cluster. The algorithm continues with the next unprocessed point, to explore another cluster. The procedure is terminated when all input points have been processed.
We discuss the computational complexity of algorithm \ref{alg:pcl} in \autoref{sec:analysis}.




\section{\lisco{}}\label{sec:lisco}

We present in this section \lisco{}. 
We first discuss the intuition behind the algorithm, i.e. a continuous, single-pass approach in clustering in contrast to existing batch based methods as discussed in \autoref{sec:preliminaries}. Subsequently, we focus on the challenges and trade-offs that continuous clustering introduces.


\subsection{Towards continuous clustering}
\label{sec:lisco:nosupport}

Based on the clustering requirements as introduced in Definition~\ref{def:problem}, each point $p$ reported by the \lidar{} sensor is temporarily clustered with each neighbor point $p'$ within distance $\epsilon$ from $p$. A cluster of points is eventually delivered if it contains at least \minpts{} points, otherwise its points are characterized as noise.
As discussed in \autoref{sec:pcl}, implementations such as the one provided by PCL limit the pipelining of the analysis because of processing stages that cannot be executed concurrently. More concretely, they first traverse all the points to populate a supporting data structure that facilitates finding the points in the $\epsilon$-neighborhood of each point $p$ (a kd-tree in the case of PCL) and subsequently they traverse a second time all the points to cluster them. 

\begin{figure*}[t!]
  \centering
  \includegraphics[width=0.8\linewidth]{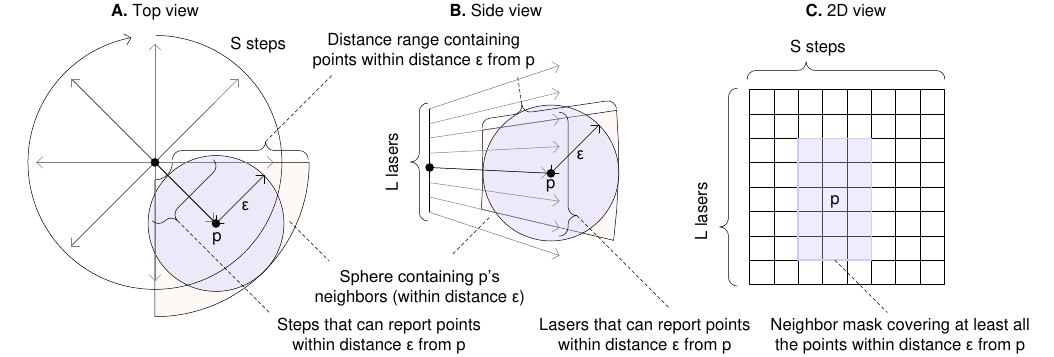}
  \caption{Top and side views (A and B) showing which steps and lasers, respectively, to include in the \textit{neighbor mask} (C) for the latter to contain at least all the points within distance $\epsilon$ from $p$.}
  \label{fig:intuition}
\end{figure*}

As shown in Figure~\ref{fig:intuition}, the intuition behind \lisco{} is that the search space for the $\epsilon$-neighborhood of point $p$ can be translated into a set of readings given by certain \emph{steps} (Figure ~\ref{fig:intuition}.A) and \emph{lasers} (Figure ~\ref{fig:intuition}.B) around $p$, and within a certain distance range from the \lidar{}'s emitting sensors. 
It should be noted that these constraints describe a region of space that may also contain points whose distance from $p$ is greater than $\epsilon$. 
Nevertheless, if points within distance $\epsilon$ from $p$ exist, they will be returned by one of these steps and lasers and will fall within the given distance range, as we further explain in~\autoref{sec:analysis}.

Thus, in order to discover the $\epsilon$-neighborhood of point $p$, by leveraging the sorted delivering of tuples from the \lidar{} sensor step after step, it is enough to explore a \textit{neighbor mask} centered in $p$, as shown in Figure \ref{fig:intuition}.C.
In this way, we eliminate the need for a search-optimized data structure like kd-trees, and we allow the algorithm to process data as they are received from the sensor.



\subsection{Coping with the challenges of continuous clustering}

The continuous one-pass analysis of \lisco{} introduces several challenges that are not found in batch based approaches. Here we discuss and explain how they are addressed in \lisco{} in the following.

\subsubsection{Partial view of neighbor mask}

Algorithm \ref{alg:getNeighbors} shows how the neighbor mask is computed, i.e. the number of previous and subsequent steps $\sigma$ and the number of upper and lower lasers $\lambda$ that possibly contain points within distance $\epsilon$ from $p$.
As discussed in \autoref{sec:preliminaries}, $\Delta \alpha$ and $\Delta \theta$ refer to the minimum angle differences between two consecutive steps and lasers, respectively.

\begin{algorithm}
\small
\caption{Given point $p$, compute the number of previous steps $\sigma$ and upper and lower lasers $\lambda$ bounding at least all the points within distance $\epsilon$ from $p$.}\label{alg:getNeighbors}
\begin{algorithmic}[1]
\Procedure{\texttt{getNeighborMask}}{$p$}
	\State $\lambda \gets \lceil | \arcsin (\epsilon / p_d ) | / \Delta \theta \rceil$
	\State $\sigma \gets \lceil | \arcsin ( \epsilon / p_d ) | / \Delta \alpha \rceil $ \label{getNeighbors:steps}
	\State \textbf{return} $\lambda,\sigma$
\EndProcedure
\end{algorithmic}
\end{algorithm}

\begin{algorithm}
\small
\caption{Main loop  of \lisco{}}\label{alg:liscogeneral}
\begin{algorithmic}[1]
\State $subclusters$ = $\emptyset$
\Event{reception of step $s$}
    \For{$l \in {1,\ldots, L}$}
        \State $p$ = $M[l,s]$ \label{lisco:readpoint}
        \If{$p_d>0$} \label{lisco:if}
            \State $\lambda,\sigma$ = \texttt{getNeighborMask}($p$) 
            \State \texttt{cluster}($p$,$\sigma$,$\lambda$,$subclusters$)
        \EndIf
    \EndFor
\EndEvent
\Event{all steps processed}
    \For{$subcluster$ $\in$ $subclusters$}
        \If{size($subcluster$) $\geq \text{\minpts{}}$}
            \Return $subcluster$
        \EndIf
    \EndFor
\EndEvent
\end{algorithmic}
\end{algorithm}

The main loop of \lisco{}, Algorithm~\ref{alg:liscogeneral}, processes points in $M$, the 2D matrix of input points (described in \autoref{sec:preliminaries}), in step and laser order (Line~\ref{lisco:readpoint}).
Each point is processed only if its distance is greater than 0 (that is, if the \lidar{}'s pulse has been reflected for the point's step and laser index) (Line~\ref{lisco:if}).
Once all the points have been processed, all the clusters containing  at least \minpts{}
points are delivered. As it can be noticed, the parameter \minpts{} does not have an effect in the complexity of the algorithm, since it is only used to filter the delivered clusters at the very end of the clustering process.
We describe how clusters are discovered and managed within the function \texttt{cluster} in the following.

Given points $p_1$ and $p_2$ within distance $\epsilon$,
$p_1$'s neighbor mask will contain $p_2$ and vice versa. To avoid comparing each pair of neighboring points twice, it is enough to consistently traverse half the neighbor mask.
\lisco{} explores the half containing the $\lambda$ lasers above and below $p$ and the $\sigma$ steps on $p$'s left side. This allows for points in $p$'s step to be processed as soon as they are delivered (points on $p$'s right side are yet to be delivered upon reception of $p$'s step points). 
Take into account that, for a minority of steps, not all the points on $p$'s left side lay on columns with a lower index than $p$'s (i.e., they are not stored on the left side of the $M$ matrix). For instance, if a point in column 2 should be compared with 3 columns on the left ($\lambda=3$), then it should be compared with columns S-1, S and 1. In such a case, some comparisons must be postponed until such steps are delivered. A point $p'$ on the left side of $p$ and within distance $\epsilon$ lays on a lower index than $p$ if $0 \leq p_s - p'_s \leq \sigma$. On the other hand, if $p'_s + \sigma > S$, then $p'$ is on the left side and within distance $\epsilon$ from points in steps $1,\ldots,p_s + \sigma - S$. In both cases, the clustering semantics defined in \autoref{sec:preliminaries} require $p$ and $p'$ to be compared, as we do in Algorithm~\ref{alg:lisco1}.


\subsubsection{Continuous cluster management}
A second challenge brought by the continuous nature of \lisco{} is that subclusters evolve as more steps arrive. Hence, a cluster identified once all the points in a rotation are processed might be the union of several previously discovered subclusters, as seen in Figure~\ref{fig:subclusters}. Figure~\ref{fig:subclusters}.A shows the subclusters found when half of the points in the rotation have been processed.
In the example, 5 points have been clustered in $C_1$ and 4 points have been clustered in $C_2$. The other non-colored points have not been clustered since they had no neighbors within distance $\epsilon$.
Figure~\ref{fig:subclusters}.B shows the clusters found when all the points in the rotation have been processed. 
At this stage, the points previously clustered in different subclusters are now clustered together.
The point marked with $x$ represents the point that has one neighbor in each of the two disjoint subclusters found by \lisco{}. Once $x$ is processed, these two subclusters should be merged.

\begin{figure}[t!]
  \centering
  \includegraphics[width=0.8\linewidth]{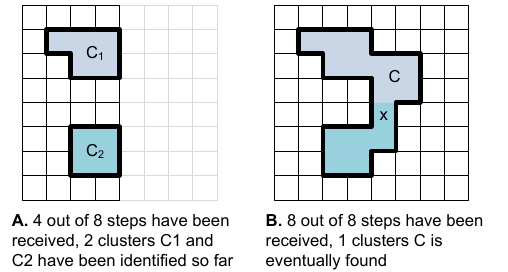}
  \caption{Example showing how the points of two subclusters, that evolve concurrently in \lisco{}'s continuous analysis, may end up in the same cluster at a subsequent step.}
  \label{fig:subclusters}
\end{figure}

Based on this observation, we introduce the following informal notions to facilitate the detailed description of \lisco{}. A \emph{subcluster} is a set of points that have been clustered together during the processing of the previously received steps of the input. A \emph{cluster} is a set of at least \minpts{} points that have been clustered together once all the steps of the input have been processed, i.e. a subcluster with cardinality at least \minpts{} is characterized as a cluster after all the steps have been processed. Finally, we consider each subcluster to have a unique identifier called \emph{head}.



Based on the above, one can notice that a subcluster can contain points that previously belonged to two or more subclusters.
Because subclusters are found continuously while the points of a rotation are being processed, the clustering algorithm requires methods to: (1)~retrieve the head of the points belonging to a subcluster and (2)~merge two subclusters together in order for the final cluster to be delivered as a single item.
In order to do this, we use in Algorithm~\ref{alg:lisco1} (overviewing the clustering process applied to each incoming point $p$) the functions \texttt{Head $H$ = getH(Point $p$)} and \texttt{merge(Head $H1$,Head $H2$)}. Once two subclusters are merged invoking function \texttt{merge}, we expect the function \texttt{getH} to return the same head for any point of the two subclusters. 
Without loss of generality, we assume this head to be $H_1$ in the following.
Because of this we define a third function \texttt{setH(point $p$,Head $H$)}. Finally, we define the function \texttt{createH()} to allow for newly discovered subclusters to be instantiated.

\begin{algorithm}
\small
\caption{Given point $p$, cluster it together with all the points already received from the \lidar{} sensor that are within distance $\epsilon$ from it.}\label{alg:lisco1}
\begin{algorithmic}[1]
\Procedure{\texttt{cluster}}{$p$,$\lambda$,$\sigma$,$subclusters$}
	\For{$p' | (0 \leq p_s - p'_s \leq \sigma \vee 1 \leq p'_s \leq p_s + \sigma - S) \wedge \left| p_l - p'_l \right| \leq \lambda \wedge \left|p_d - p'_d\right| \leq \epsilon$}
	    \State $H_1$=\texttt{getH}(p)
	    \State $H_2$=\texttt{getH}(p')
	    \If {$H_1 \neq H_2 \wedge ||p - p'||_2 \leq \epsilon $}
    	    \If{$H_1=\emptyset \wedge H_2=\emptyset$} \label{lisco1:caseA} 
    		        \State $H$ = \texttt{createH}() \label{lisco1:newStart}
    		        \State \texttt{setH}($p$,$H$)
    		        \State \texttt{setH}($p'$,$H$) \label{lisco1:newEnd}
    		        \State $subclusters$.\texttt{add}($H$)
    	    \ElsIf{$H_1 = \emptyset \wedge H_2 \neq \emptyset$} \label{lisco1:caseB}
    	        \State \texttt{setH}($p$,$H_2$)
            \ElsIf{$H_1 \neq \emptyset \wedge H_2 = \emptyset$} \label{lisco1:caseC}
                \State \texttt{setH}($p'$,$H_1$)
    		\Else  \label{lisco1:caseD}
    		    \State $subclusters$.\texttt{remove}($H_2$)
    		    \State \texttt{merge}($H_1$,$H_2$))
    		\EndIf
		\EndIf
	\EndFor
\EndProcedure
\end{algorithmic}
\end{algorithm}

As shown in Algorithm~\ref{alg:lisco1}, four different cases should be checked for two points within distance $\epsilon$ that do not belong to the same subcluster:
\begin{itemize}
    \item \textit{Line \ref{lisco1:caseA}}: None of the two points belongs to a subcluster. 
    In this case, a new subcluster head is created and set for both points
    \item \textit{Line \ref{lisco1:caseB}}: Point $p$ does not belong to a subcluster while point $p'$ does. 
    In such a case, point $p$ will refer to the same head as point $p'$.
    \item \textit{Line \ref{lisco1:caseC}}: Point $p$ belongs to a subcluster while point $p'$ does not. 
    In such a case, point $p'$ will refer to the same head as point $p$.
    \item \textit{Line \ref{lisco1:caseD}}: Both points $p$ and $p'$ have been clustered but to different subclusters. In such a case these two subclusters are merged together.
\end{itemize}

\section{Algorithmic implementation}\label{sec:implementation}

We discuss the details of the algorithmic implementation of \lisco{} in this section.

Data points are kept in a 2D matrix, $M$.
The number of rows and columns in the matrix is equal to the number of lasers and steps respectively.
Upon reception of a column of points from \lidar{}, which contains the reflected points of all lasers in one step, we store them in the corresponding column of the matrix. By using the laser and the step number, all the attributes of a point can be extracted in constant time. Each entry in $M$ holds the attributes of the corresponding point and a pointer (initially set to NULL) to the \textit{head of its subcluster}. 
The \emph{head} of a subcluster is defined as the point with the lowest indices in lexicographical order of steps and lasers during the creation of a new subcluster.
When two subclusters are merged, the head from the subcluster with the largest number of members is maintained.

$Subclusters$ is a \emph{hash map} used in order to keep track of subclusters and their corresponding members.
It is implemented as a linked list of arrays, where each key is the header of a subcluster and its members are stored in the array. 
If the size of a subcluster exceeds the size of the array, a new array is linked to the tail of the current array, so that subclusters can grow without restriction. At the end of the clustering procedure, we use $subclusters$ to traverse through subcluster heads. Each subcluster that has more than \minpts{} members is announced as a cluster, otherwise it is characterized as noise.

To keep  of \lisco{}'s time complexity low (discussed in \autoref{sec:analysis}) we aim at efficient time complexity of the main methods used in our algorithms.
As shown in Algorithm~\ref{alg:getNeighbors}, function \texttt{getNeighborMask} is executed in constant number of steps, since it boils down to a fixed number of numerical operations.
Similarly, functions \texttt{createH} and \texttt{setH} can be also implemented to incur in $O(1)$ complexity, as we discuss in the following.
The algorithmic implementation of \texttt{getH} and \texttt{merge} induces the following trade-off.
On the one hand, \texttt{merge} can be implemented to induce $O(1)$ time-cost; this can be done  by maintaining a hierarchy of subclusters being part of the same subcluster while incurring a higher cost for the \texttt{getH}, linear in the number of subclusters.
Figure~\ref{fig:approachB} shows how some of the points clustered together once all the data is processed point to head $H_1$ via head $H_2$. For these points the \texttt{getH} method has a cost higher than that of the points directly pointing to $H_1$, which depends on the chains induced by the data structure to maintain the hierarchy.
In the proposed implementation we opt for an $O(1)$ cost for method \texttt{getH} and a higher cost for \texttt{merge}, as seen in~Figure~\ref{fig:approachA} and~\autoref{sec:analysis}.
The reason, as can be seen in Algorithm~\ref{alg:lisco1} and based also on our empirical evaluation in \autoref{sec:evaluation},  is that \texttt{getH} is executed twice for pairs of points being compared, while \texttt{merge} is executed significantly less often.

\begin{figure}[t!]
  \centering
  \includegraphics[width=0.8\linewidth]{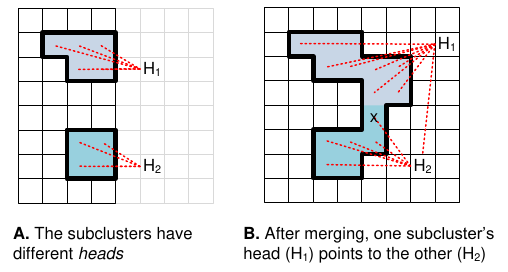}
  \caption{Possible implementation in which the \texttt{merge} method is made $O(1)$ by hierarchically linking heads of subclusters belonging to the same subcluster. Notice that the complexity of \texttt{getH} is no longer $O(1)$ but linear in the number of subclusters for some of the points.}
  \label{fig:approachB}
\end{figure}

\begin{figure}[t!]
  \centering
 \includegraphics[width=0.8\linewidth]{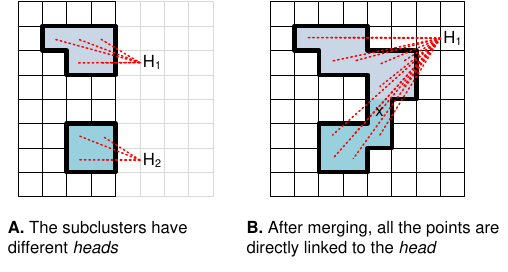}
  \caption{Possible implementation in which the \texttt{getH} method can run in $O(1)$ steps by having all points directly linked with the subcluster head. Notice the head of all the points of subcluster $C_2$ has been updated when the subcluster has been merged with subcluster $C_1$.}
  \label{fig:approachA}
\end{figure}

In detail, the implementations of the functions used in Algorithm~\ref{alg:lisco1} are as follows:\\
- \texttt{createH}: This function gets two points that do not belong to any subcluster, and returns the one with the lower index-pair for step and laser. Since the returned point will be head of a subcluster, a new node is created in the $subclusters$ and the head is mapped to it.\\
- \texttt{setH}: This function sets the pointer of a point to the head point. By calling this function, we are adding a point to a subcluster with an identified head. The point also needs to be added as the last element in the array of the mapped head.\\
- \texttt{getH}: This function reads the pointer to get the head point. If two points are in the same subcluster, they get the same head point as the result of this function.\\
- \texttt{merge}: If two points belong to different subclusters, \lisco{} merges the two subclusters by calling this function. It chooses the subcluster with bigger number of members as the base subcluster, and merges the other one by changing the head of its members to the base subcluster's head. After merging subclusters, it is also necessary to remove the merged subcluster's head from  the $subclusters$ hash map and append its array to the base subcluster's array.
\section{Analysis}
\label{sec:analysis}
\subsection{Correctness}

Based on \lisco{}'s functionality and algorithmic implementation, we discuss here why \lisco{}'s outcome satisfies Definition~\ref{def:problem}.

\begin{claim}\label{mainclaim}
If two points $p$ and $p'$ are in the $\epsilon$-neighborhood of each other they will either be in the same cluster at the end of \lisco{} procedure or be characterized as noise, in the same way as given in Definition~\ref{def:problem}.
\end{claim}

\begin{proof} (sketch) Consider that, w.l.o.g. $p$ is processed second. As argued in paragraph \textit{getNeighborMask} in the previous section, $p'$ will be found that it belongs in the $\epsilon$-neighborhood of $p$. This implies that they will be merged/inserted in the same subcluster. Unless that subcluster in the end is found to contain fewer points than $MinPts$, it will be return as a final cluster in the end of the main loop of \lisco{}.
\end{proof}

\subsection{Complexity}
In this section, we discuss the complexity analysis of \lisco{} and compare it with \PCLeCluster{}.


Regarding \PCLeCluster{}, the required processing work volume is similar to the DBSCAN algorithm, i.e. building a spatial index (kd-tree) and using it to execute region queries for each point, resulting in an overall \emph{expected time complexity} of $O(nlogn)$ processing steps \cite{wald2006building,ester1996density,patwary2012new}.

\remove{ As discussed in \autoref{sec:pcl}, \PCLeCluster{} has two main parts; building the kd-tree, and clustering data points. 
Building a kd-tree on $n$ data points has a time complexity of O($n log^2 n$) or even O($n^2$). But, it is possible to reduce the building complexity to O($n log n$) \cite{}. To analyze the clustering part of \PCLeCluster{} we need to know the time complexity of a query region in a kd-tree.
As discussed in \cite{ester1996density}, time complexity of a query region for one point with small $\epsilon$ compared to the size of the whole dataset is O(log n). Therefore, the overall average time complexity for \PCLeCluster{} is O(n log n).
} 

\begin{claim}\label{complexityclaim}
\lisco{}'s time complexity is linear in the number of points, multiplied by a factor that depends on the size of the clusters in the set of data points.
In the \emph{worst-case} where there is a big cluster of $O(n)$ points, it can take $O(nlogn)$ processing steps for \lisco{} to complete.
\end{claim}

\begin{proof} (sketch)
Overall, the time complexity of \lisco{} is the number of iterations in the main loop (i.e. $n$, as the number of points), times the work in each iteration, i.e. for each point (i)~finding its $\epsilon$-neighborhood, and (ii)~working with each point in the neighborhood.

Part (i) from above, induces an asymptotically constant cost, depending on $\epsilon$, as it is performed through the comparisons implied by the masking operation \textit{getNeighborMask}. Part of the  $\epsilon$-neighborhood of a point $p$, i.e. the points with smaller step index, is compared with $p$ through step 2 of Algorithm~\ref{alg:lisco1} on behalf of $p$, while for each of the remaining points $p'$ in p's $\epsilon$-neighborhood, $p$ will be identified  as part of the $\epsilon$-neighborhood of $p'$ when the respective step is executed on behalf of $p'$. 

Regarding part (ii), functions' \texttt{createH} and \texttt{setH} algorithmic implementation incurs a constant number of processing steps each, as we explain in \autoref{sec:implementation}.
Moreover, as explained in the aforementioned section, \texttt{getH}  induces  $O(1)$ time-cost, as a point can identify its head in $O(1)$ (e.g. with a direct link). This incurs a cost that is $O(x)$ for \texttt{merge}, where $x$ is the size of the smaller subcluster, since the merging itself needs to update the head for all the points of the smaller one of the subclusters being merged.

Since the merge function chooses the subcluster with the bigger number of points as the base subcluster, in the worst-case  the clustering has a huge subcluster of $O(n)$ points, and  an unlikely scenario for constructing it, might require \lisco{} to merge roughly equally-sized subclusters at each of the merge operations leading to the big subcluster -- any other combination of subclusters would lead at most to an equal cost as the one described. Since halving $O(n)$ points can be made at most $O(logn)$ times, we can observe that the worst-case \textit{total} number of merge-related processing steps will be dominated by a sequence of $O(n log n)$ steps, which will be the dominating cost in the worst case complexity of \lisco{}.
\end{proof}



\remove{
On the other hand, as discussed in \autoref{sec:lisco}, \lisco{} utilizes the data as delivered from the \lidar{} sensor, so there is no need to build a tree. The processing cost of each point (Algorithm~\ref{alg:lisco1}) depends on the cost of functions \texttt{getNeighborMask}, \texttt{createH}, \texttt{getH}, \texttt{setH} and \texttt{merge}. 
As shown in Algorithm~\ref{alg:getNeighbors}, function \texttt{getNeighborMask} is executed in constant number of steps, since it boils down to a fixed number of numerical operations.

Similarly, functions \texttt{createH} and \texttt{setH} can be also implemented to incur in $O(1)$ complexity, as we discuss in \autoref{sec:implementation}.

However, \texttt{getH} and \texttt{merge} cannot both have both $O(1)$ complexity at the same time.
On one hand, \texttt{getH}  can be made O(1), as a point can identify its head in O(1) (e.g. with a direct link). This nevertheless incurs a cost that is O(x), where x is the size of a subcluster, since the merging itself needs to update the head for all the points of one of the subclusters being merged. 
As we show in Figure~\ref{fig:approachA} (building on the previous example in Figure~\ref{fig:subclusters}) the points pointing to head $H_2$ in Figure~\ref{fig:approachA}.A are pointing to $H_1$ in Figure~\ref{fig:approachA}.B, which implies that the \texttt{setH} method has been invoked for all of them.
On the other hand, \texttt{merge} can be made O(1) by maintaining a hierarchy of subclusters being part of the same subcluster while incurring a higher cost for the \texttt{getH}, linear in the number of subclusters.
Figure~\ref{fig:approachB} shows how some of the points clustered together once all the data is processed point to head $H_1$ via head $H_2$. For these points the \texttt{getH} method has a cost higher than that of the points directly pointing to $H_1$.
} 

\section{Evaluation}\label{sec:evaluation}

In this section, we present our experimental methodology and results of \lisco{} and compare them with those of \PCLeCluster{} algorithm.
Since the clustering outcomes of \PCLeCluster{} and \lisco{} are the same, we do not need to compare the clusters and we can focus on the completion time for each approach.

\subsection{Evaluation setup}
To run \PCLeCluster{} we use the $EuclideanClusterExtraction$ class from PCL library, which is designed to cluster 3D point clouds and is implemented in C++. 
We also implemented \lisco{} in C++11 and compiled both of algorithms with gcc-4.8.4 using the -O3 optimization flag.
All the experiments have been run on the same system running Linux with 2.00GHz Intel(R) Xeon(R) E5-2650 processor and 64GB Ram.

\subsection{Data}
We used both synthetic and real-world datasets.
The real-world dataset has been collected from the \ford{} dataset \cite{pandey2011ford} and the synthetic ones have been generated using the Webots simulator \cite{michel2004cyberbotics}.
We use synthetic datasets to explore the effect of data (e.g. different total number of points that have been collected by the \lidar{}, different densities and distances of objects) on the performance of algorithms.

There are five scenarios for synthetic datasets. SCEN1 and SCEN2 have the same and few number of objects but we changed the position of the objects to near and far from \lidar{}. Near objects reflect more points while far objects reflect fewer points with a larger gap between two nearby points. Similarly, SCEN2 and SCEN3 have the same number of objects (number of objects are more than previous scenarios) and we only changed the position of the objects. Finally, SCEN5 represents a high density environment with a lot of objects. Figure \ref{fig:scenario} shows five simulated environments for several scenarios.
In all the environments, a VelodyneHDL64E is used to collect data points and generate a dataset within one physical rotation.

Table \ref{table:structure} summarizes the properties of the synthetic datasets.
Since we used the same specifications for the \lidar{} in all scenarios, the number of steps and lasers for one physical rotation is the same, so the total number of points (including NULL points and ground points) is the same and it is equal to 72000.
After removing the ground points and eliminating NULL points, we get a number of reflected points.
As shown in the table, with the same number of objects in the environment (e.g. SCEN1 and SCEN2), if we change the position of objects to near or far, we get different numbers of reflected points.

\begin{figure*}[t!]
\centering
\subfigure[SCEN1 - The simulated sparse environment in which objects are close to the \lidar{} which is located on the black car]{\includegraphics[width=0.19\textwidth,height=3cm]{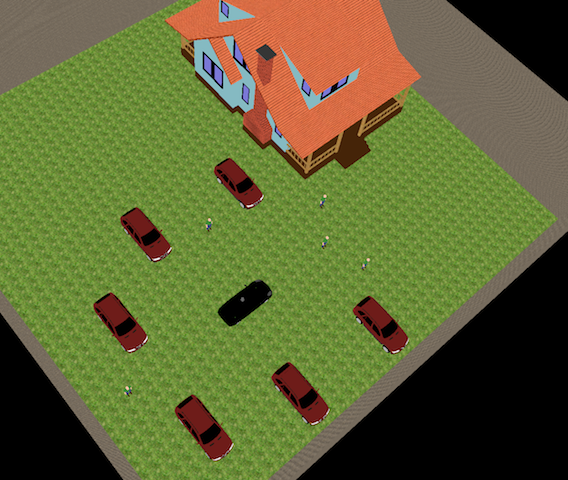}}
\subfigure[SCEN2 - The simulated sparse environment in which objects are far from the \lidar{} which is located on the black car]{\includegraphics[width=0.19\textwidth,height=3cm]{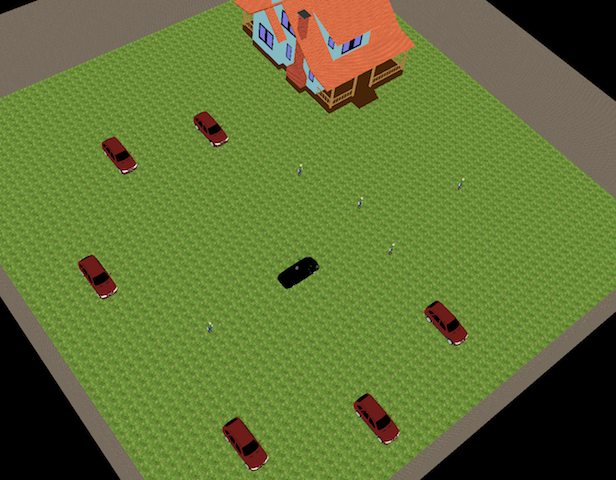}}
\subfigure[SCEN3 - The simulated dense environment in which objects are close to the \lidar{} which is located on the black car]{\includegraphics[width=0.19\textwidth,height=3cm]{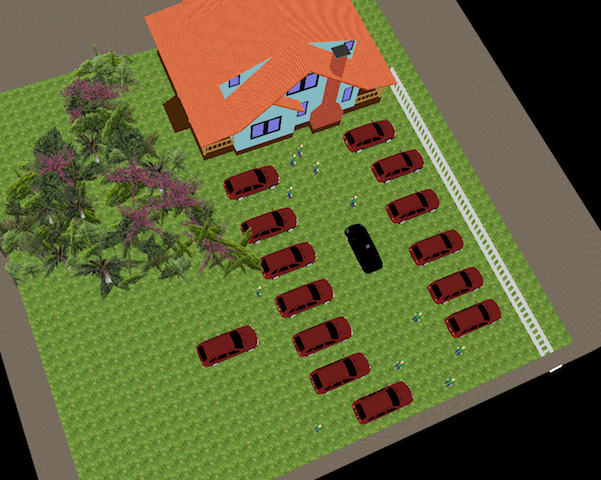}}
\subfigure[SCEN4 - The simulated dense environment in which objects are far from \lidar{} which is located on the black car]{\includegraphics[width=0.19\textwidth,height=3cm]{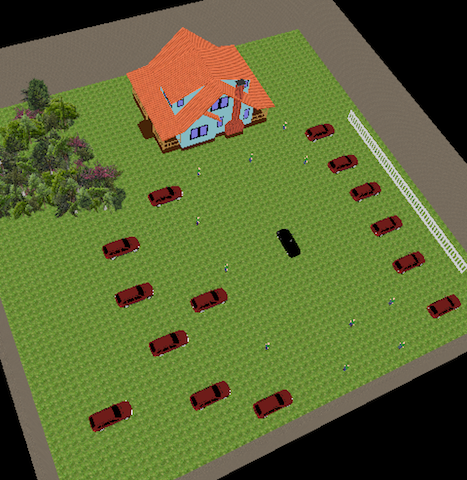}}
\subfigure[SCEN5 - The simulated room for high density environment. The \lidar{} is located on the purple column]{\includegraphics[width=0.19\textwidth,height=3cm]{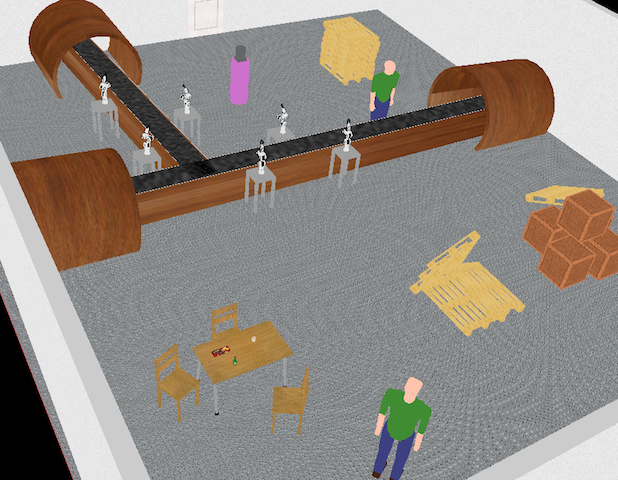}}
\caption{Different scenarios for simulated environments.}
\label{fig:scenario}
\end{figure*}

\begin{table}[!h]
\centering
 \begin{tabular}{|c c |}
 \hline
Name & \# Points after removing NULL points and ground points\\
 \hline
 SCEN1 & 26891 \\
 SCEN2 & 16218 \\
 SCEN3 & 39028 \\  
 SCEN4 & 18229 \\
 SCEN5 & 64518 \\
 \hline
\end{tabular}
 \caption{Properties of synthetic datasets and the effect of number of objects and their distances from \lidar{} on number of points after ground removal. }
\label{table:structure}
\end{table}

\subsection{Performance evaluation}
The execution time is measured from the time-instant the first data point of the dataset is received until the time instant when the clustering algorithm has processed all data points of one full physical rotation.
A higher value of $\epsilon$ implies a larger $\epsilon$-neighborhood of a point, hence the clustering algorithm needs more time to search in the neighborhood.

\subsubsection{Synthetic datasets}
Figure \ref{fig:compare} shows the average execution time with confidence level 99\% on 20 runs with different values of $\epsilon$ and constant value of \minpts{} = 10.
Since the maximum margin of error for a confidence level 99\% is small, we can not distinguish them clearly in the figure.
We chose a range $[0.1-1]$ meters for $\epsilon$, so that for example if $\epsilon = 0.4$, all the objects that their closest points have at least 40 centimetres distance from each other, should be detected as separated objects.
While  clustering with smaller values of $\epsilon$ find at least one cluster for each object, bigger values increase the probability of clustering distinct objects together. For example, with $\epsilon=1$ for SCEN3, all the cars at each side of the black car, are clustered together which leads to incorrect segmentation.

\begin{figure*}[t!]
  \centering
  \includegraphics{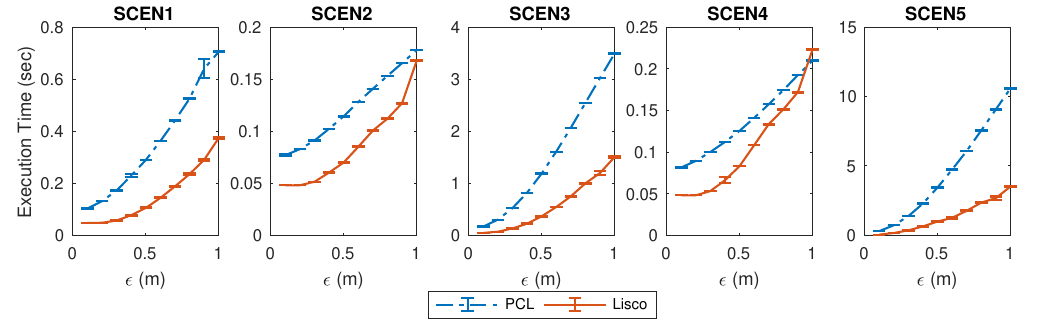}
  \caption{The average execution time on synthetic datasets with confidence level 99\% over 20 runs for \minpts{} = 10 and different values of $\epsilon$}
  \label{fig:compare}
\end{figure*}

As expected, by increasing the value of the $\epsilon$, the execution time increases for both algorithms. As it can be seen, when the number of points is high, regardless of the value of the $\epsilon$, \lisco{} is always faster than PCL. Only for a dataset with a relatively small number of points and when $\epsilon$ is set to a high value PCL has slightly better performance than \lisco{} (figure \ref{fig:compare} SCEN2 and SCEN4). The effect of number of points is also shown in figure \ref{fig:comparePoints}. As discussed in \autoref{sec:analysis} building a kd-tree and using it to find nearest neighbors in PCL becomes a bottleneck when the number of points is high.

\begin{figure}[t!]
  \centering
  \includegraphics{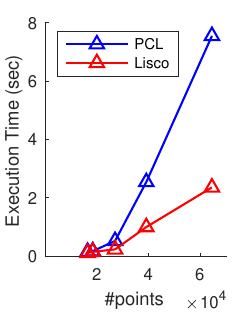}
  \caption{Scalability of \lisco{} and PCL with respect to the number of points}
  \label{fig:comparePoints}
\end{figure}

\subsubsection{Real-world dataset}
This \ford{} dataset is collected by an autonomous ground vehicle testbed, with a Velodyne HDL-64E \lidar{} scanner \cite{pandey2011ford}. The vehicle path trajectory in this dataset contains several large and small objects (e.g. buildings, vehicles, pedestrians, vegetation, etc.).
We have tested \PCLeCluster{} and \lisco{} on 2280 rotations of this dataset and compare their execution times with confidence level 99\%.

Figure \ref{fig:result} shows the results of the comparison for $\epsilon$ values 0.3, 0.4, and 0.7.
Among all the rotations, the minimum number of reflected points after ground removal is 5000, the maximum is 75550, and the average is 50225.
As shown in the figure, \lisco{} outperforms \PCLeCluster{} in real-world datasets. In real-world data, generally there are more objects around the \lidar{} and therefore there are more reflected points besides the ground. Since \lisco{} processes points upon receiving them from \lidar{}, it can save more time and it has thus better throughput.

\begin{figure}[t!]
  \centering
  \includegraphics{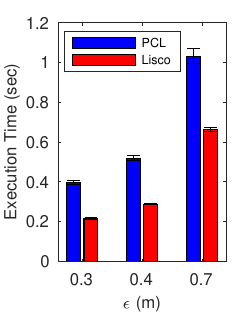}
  \caption{The average execution time of running PCL and \lisco{} over 2280 rotations of real-world dataset for \minpts{} = 10 and different values of $\epsilon$}
  \label{fig:result}
\end{figure}
\section{Other related work}
\label{sec:related}

Data clustering has been studied for several decades and existing algorithms have been categorized into four classes: density-based, partition-based, hierarchy-based, and grid based \cite{han2011data}. Due to their ability in finding arbitrarily shaped clusters without requiring to know the number of the clusters a priori, density-based methods are widely used in different applications. Well-known algorithms of this class include DBSCAN \cite{ester1996density} and OPTICS \cite{ankerst1999optics}. However, to use these algorithms in big data applications and overcome their performance bottleneck in dealing with extremly large datasets, there are several attempts to parallelize DBSCAN \cite{patwary2012new, kumari2017exact}. In parallel models, the clustering procedure is divided into three steps: 1) data distribution (e.g. using kd-tree) 2) local clustering (which is splitted on several machines) 3) merging of local clusters. Although the efficiency is improved by splitting the clustering on several machines, pipelining of steps has not being studied yet.

Rusu et al.~\cite{rusu2009close} introduce an Euclidean-distance-based clustering which is a partition-based clustering method that produces arbitrarily shaped clusters. This approach is designed for unorganized data points. So, to facilitate searching for nearest neighbors, first a kd-tree is built over the dataset and then clustering is being performed. In other works~\cite{woo2002new, vo2015octree}, an octree is used to identify the neighbors before starting the clustering procedure.

Since \lidar{} data points are implicitly ordered, organizing them (e.g. in a tree) may be avoided, similar to the spirit of this paper.
Specifically, the characteristics of the sensor data can be used to establish neighborhood relations between points~\cite{klasing2009realtime, moosmann2009segmentation}.
Klasing~\cite{klasing2009realtime} et al. proposed a clustering method for 2D laser scans that rotate with an independent motor to cover a 3D environment. 
While the proposed method compares points across different scans similarly to the problem studied in this work, the semantics of Definition~\ref{def:problem} are not enforced resulting to a lower accuracy.
Moosmann et al.~\cite{moosmann2009segmentation} proposed an approach to turn the scan into an undirected graph to retrieve the neighborhood information of each point during clustering, but they have not studied pipelining building the graph and clustering.
Zermas et al.~\cite{zermas2017fast} recently proposed a clustering method specific to the structure of \lidar{} data points. This approach processes one \texttt{scan-line} (a layer (rotation) of points that are produced from the same laser) at a time and merges nearby clusters from different scan-lines. However, the entire rotation is needed as the algorithm does more than one pass over the data. Also, this approach, similarly to previous works, still relies on a kd-tree for some necessary nearest neighbor searches.
Moreover, the neighborhood criterion for points clustered in the same scan-line does not take into account the distance of the point from the sensor, and thus does not guarantee the semantics of Definition~\ref{def:problem}.

Clustering \lidar{} data points are being used in wide range of applications \cite{li2012new, klasing2008clustering, sampath2010segmentation}.
Among all, autonomous vehicle applications are one of the most challenging since they need fast and accurate results \cite{wang2012could, himmelsbach2010fast, douillard2011segmentation, zermas2017fast}.
In \cite{douillard2011segmentation}, a set of voxelisation and meshing segmentation methods are presented. Wang et al.\cite{wang2012could} first separates data into foreground and background. then a clustering procedure is conducted only on the foreground segments.


\label{sec:relatedwork}
\section{Conclusions and Future Work}
\label{sec:conclusions}

This work is about one of the challenges in common big data applications, namely
leveraging the  information carried by high-rate streams
through  efficient methods that can rapidly distill
the valuable information from the raw measurements. A
common problem in the analysis of LiDAR sensor data, that generate date at rates of megabytes per second,
is clustering of the raw distance measurements, in order to facilitate detection of objects surrounding the sensor.

\lisco{} represents a streaming approach to process the \lidar{} points while the data is being collected. This characteristic helps  to facilitate extraction of clusters in a continuous fashion and contribute to real-time processing.  By keeping track of the different subcluster heads, \lisco{} can deliver subclusters to the user at anytime by request, i.e. provide continuous information.
\remove{It is necessary to consider that delivering does not mean that subclusters are final, but they are subclusters that clustering algorithm could find so far and they can be changed after collecting more points.
The final clusters can be delivered once all the data is processed.}

Important follow-up questions include the parallelization of \lisco{}'s processing pipeline to take advantage of computing architectures for the corresponding deploy environments. 
This necessitates algorithmic implementations in a variety of processing architectures,  such as manycores/GPUs, SIMD systems, single board devices and high-end servers, to explore \lisco{}'s properties in  a broad range of cloud and fog architectures and evaluate its impact on applications that can be deployed on such systems.
In addressing such questions it will be useful to leverage the benefits of efficient fine-grained synchronization methods in streaming-centered and bulk-operations-enabled data structures, as proposed in \cite{cederman2013concurrent,Gulisano2015debs,gulisano2016scalejoin,Nikolakopoulos2016ericsson}. 


\end{document}